\documentclass{article}

\usepackage{amsthm}

\usepackage{hyperref}

\usepackage{tikz}
\usetikzlibrary{calc}
\tikzstyle{every node}=[circle, fill=black, inner sep=0pt, minimum width=2.5pt]

\newtheorem{define}{Definition}[section]
\newtheorem{lemma}[define]{Lemma}
\newtheorem{cor}[define]{Corollary}
\newtheorem{claim}[define]{Claim}

\hypersetup{
  linkcolor=green,        
  citecolor=green,        
  colorlinks=false,       
}

\begin{document}

\title{Polarity and Monopolarity of $3$-colourable comparability graphs}

\author{
  Nikola Yolov
  \footnote{The author is funded by
    the Engineering and Physical Sciences Research Council (EPSRC) Doctoral Training Grant
    and the Department of Computer Science, University of Oxford}\\
  \texttt{nikola.yolov@cs.ox.ac.uk}
}
\maketitle

\begin{abstract}
  We sharpen the result that polarity and monopolarity are
  NP-complete problems
  by showing that they remain NP-complete if the input graph
  is restricted to be a $3$-colourable comparability graph.

  We start by presenting a construction reducing
  $1$-$3$-SAT to monopolarity of $3$-colourable comparability graphs.
  Then we show that polarity is at least as hard as monopolarity
  for input graphs restricted to a fixed disjoint-union-closed class.
  We conclude the paper by stating that both polarity and monopolarity
  of $3$-colourable comparability graphs are NP-complete problems.

\end{abstract}

\section{Introduction}
A partition $(A, B)$ of the vertices of a graph $G$ is called \emph{polar}
if $\overline{G[A]}$ and $G[B]$ are unions of disjoint cliques.
A polar partition $(A, B)$ is called \emph{monopolar} if $A$ is an independent set,
and \emph{unipolar} if $A$ is a clique.
A graph $G$ is called to be \emph{polar}, \emph{monopolar} or \emph{unipolar}
if it admits a polar, a monopolar or a unipolar partition respectively.
If $(A, B)$ is a polar partition of $G$, then $(B, A)$ is polar partition of $\overline{G}$,
hence a graph is polar iff its complement, $\overline{G}$, is polar.
First studied by Tyshkevich and Chernyak in
\cite{tyshkevich1985decomposition} and \cite{tyshkevich1985algorithms},
monopolar graphs are a natural generalisation of bipartite and split graphs,
and polar graphs are a generalisation of monopolar and co-bipartite graphs.
The problems of deciding whether a graph is polar, monopolar and unipolar
are called \emph{polarity}, \emph{monopolarity} and \emph{unipolarity} respectively.
Polarity \cite{polar_NP-hard}
and monopolarity \cite{monopolar_NP-hard} are NP-complete problems
and enjoyed a lot of attention recently.
In this work we sharpen the hardness result by
showing that the problems remain NP-complete even if the input graph is restricted
to be a $3$-colourable comparability graph.
In contrast, unipolarity can be resolved in quadratic time~\cite{recognition_unipolar}.

It is shown in \cite{3-c-NP-hard} that
polarity and monopolarity of $3$-colourable graphs are NP-complete problems.
On the other hand, there are polynomial time algorithms for
monopolarity and polarity of permutation graphs \cite{polar_perm},
and a polynomial time algorithm for monopolarity
of co-comparability graphs \cite{polar_co_comparability}.
Note that a graph is a permutation graph if and only if it is both
a comparability and a co-comparability graph,
hence the latter algorithm is a generalisation of the prior, but less efficient.
It is natural to ask whether there is a polynomial time polarity or monopolarity
algorithm for the other superclass of
permutation graphs -- comparability graphs.
In this paper we give a negative answer (provided P $\neq$ NP),
in fact we show that even if the graph is restricted to a smaller class
-- the class of $3$-colourable comparability graphs -- the problem remains NP-hard.
Polarity of comparability graphs and polarity of co-comparability graphs
are easily seen to be polynomially reducible to each other,
since a graph is polar iff its complement is polar,
hence we can deduce that polarity of co-comparability graphs is also NP-complete. 
The table below summarises our discussion so far.
(Here $C$ and $C^c$ are used to denote the classes of comparability
and co-comparability graphs respectively.)

\paragraph{}
\begin{tabular}[pos]{l | l | l}
  Input graph               & Monopolarity & Polarity \\
  \hline
  $3-$colourable            & NP-c \cite{3-c-NP-hard}         & NP-c \cite{3-c-NP-hard} \\
  $C \cap C^c$              & P \cite{polar_perm}             & P \cite{polar_perm} \\
  $C^c$                     & P \cite{polar_co_comparability} & NP-c, this paper \\
  $C$                       & NP-c, this paper                & NP-c, this paper \\
  $3-$colourable $\cap$ $C$ & NP-c, this paper                & NP-c, this paper \\
  \hline
\end{tabular}
\paragraph{}

We note that Churchley announced an unpublished work proving
that monopolarity of comparability graphs is NP-complete.
This paper is independent of it.
In comparison, we address both polarity and monopolarity,
and consider a smaller class -- the class of $3$-colourable comparability graphs,
in contrast to the class of comparability graphs.

In order to show the polarity result
we present a brief lemma stating that polarity is not easier than monopolarity
for graph classes closed under disjoint union.
The question whether there exists a class for which polarity is easier than monopolarity
was asked in \cite{3-c-NP-hard}.
We contribute to the answer by stating that such class
is certainly not closed under disjoint union.

\section{Definitions and Notation}
We use standard notation for graph theory and highlight that $G[S]$
is used to denote the induced subgraph of $G$ by the vertices $S \subseteq V(G)$.
A \emph{union of disjoint cliques} is a graph whose vertices can be partitioned into blocks,
so that two vertices are joined by an edge if and only if they belong to the same block.
A \emph{co-union of disjoint cliques} is the complement of such graph.

A graph $G$ is \emph{$k$-colourable} if its vertices can be covered by $k$ independent sets.
It is NP-hard to decide whether a graph is $k$-colourable
for $k \ge 3$~\cite{michael_johnson}.
A graph $G$ is a comparability graph if each edge can be oriented towards one of its
endpoints,
so that the result orientation is transitive.
There is an algorithm to test if a graph is a comparability graph
in $O(\mathrm{MM})$ time~\cite{spinrad},
where MM is the time required for a matrix multiplication.
A graph is \emph{perfect} if $\chi(H) = \omega(H)$ for every induced subgraph $H$.
We note that there is an interesting connection
between unipolar and perfect graphs -- almost all perfect graphs
are either unipolar or co-unipolar \cite{promelsteger}.

Comparability graphs are easily seen to be perfect,
and therefore a $3$-colourable comparability graph is simply
a $K_4$-free comparability graph.
The complete graph $K_4$ is an example of a comparability
graph which is not $3$-colourable,
and $C_5$ is an example of a $3$-colourable graph which is not comparability,
hence $3$-colourable comparability graphs is a proper subset of the classes above.

\section{Monopolarity}
It this section we show that the problem of deciding whether a
$3$-colourable comparability graph is monopolar is NP-complete.
We call a CNF-formula $\phi = \bigwedge_i C_i$ a \textbf{positive} $3$-CNF formula,
if each clause $C_i$ contains exactly three non-negated literals.
We transform the following NP-complete problem ($1$-$3$-SAT) \cite{michael_johnson}
into the problem above:

\emph{Instance:}
A positive 3-CNF-formula $\phi$ with variables $c_1 \ldots c_n$.

\emph{Question:}
Is there an assignment of the variables $\{c_i\} \rightarrow \{true, false\}$,
such that each clause contains exactly \textbf{one} true literal.

\paragraph{}
Recall that a partition $(A, B)$ of $G$ is monopolar if
$A$ is an independent set and $G[B]$ is a union of disjoint cliques.
For a fixed monopolar partition $(A, B)$,
we call a vertex $v$ \emph{left} if $v \in A$, and \emph{right} otherwise.
(Imagine that a partition $(A, B)$ is always drawn
with $A$ on the left-hand side and $B$ on the right-hand side).
Observe that if $v$ is a left vertex, the neighbours of $v$ are right,
hence they induce a union of disjoint cliques.

Consider the graph $Q$ in Figure 1.
\begin{figure}[h]
  \centering
  \begin{tikzpicture}
	\node(v_1) at (0, 0) [label = left:$v_1$] {};
	\node(v_2) at (1, 0) [label = right:$v_2$] {};
	\node(v_3) at (1/2, {1 / sqrt(2)}) [label = left:$v_3$] {};
	\node(v_4) at (1/2, {- 1 / sqrt(2)}) [label = left:$v_4$] {};
	\node(v_5) at (2, 0) [label = right:$u$] {};
	\draw (v_1) -- (v_2);
	\draw (v_1) -- (v_3);
	\draw (v_1) -- (v_4);
	\draw (v_2) -- (v_3);
	\draw (v_2) -- (v_4);
	\draw (v_5) -- (v_3);
	\draw (v_5) -- (v_4);
\end{tikzpicture}
  \caption{$Q$}
\end{figure}

\begin{claim}
  The only monopolar partition of $Q$ is $(\{v_3, v_4\}, \{v_1, v_2, u\})$.
\end{claim}
\begin{proof}
  The neighbourhood of $v_1$ and $v_2$ is $P_3$
  (a path with three vertices), hence they must be both right
  for all monopolar partitions of $Q$.
  As $v_3$ and $v_4$ are not connected and they share a right neighbour,
  at least one of them is left.
  Therefore $u$ has a left neighbour, so $u$ must be right.
  However, if $v_1$, $v_2$ and $u$ are right, then $v_3$ and $v_4$ must be left.
\end{proof}

Consider the graph $H$ in Figure~2.
\begin{figure}[h]
  \centering
  \begin{tikzpicture}
	\node(v_1) at (0, 0) [label = left:$v_1$] {};
	\node(v_2) at (1, 0) [label = right:$v_2$] {};
	\node(v_3) at (1/2, {1 / sqrt(2)}) [label = left:$v_3$] {};
	\node(v_4) at (1/2, {- 1 / sqrt(2)}) [label = left:$v_4$] {};
	\node(v_5) at (2, 0) [label = above:$v_5$] {};
	\draw[->] (v_1) -- (v_2);
	\draw[->] (v_1) -- (v_3);
	\draw[->] (v_1) -- (v_4);
	\draw[->] (v_2) -- (v_3);
	\draw[->] (v_2) -- (v_4);
	\draw[->] (v_5) -- (v_3);
	\draw[->] (v_5) -- (v_4);

	\node(t_1) at (3, {1 / sqrt(2)})  [label = above:$t_1$] {};
	\node(t_2) at (3, 0)              [label = above:$t_2$] {};
	\node(t_3) at (3, {-1 / sqrt(2)}) [label = above:$t_3$] {};
	\draw[->] (v_5) -- (t_1);
	\draw[->] (v_5) -- (t_2);
	\draw[->] (v_5) -- (t_3);

	\node(v_6) at (4, {1 / sqrt(2)})  [label = above:$v_6$] {};
	\node(v_7) at (4, 0)              [label = above:$v_7$] {};
	\node(v_8) at (4, {-1 / sqrt(2)}) [label = above:$v_8$] {};
	\draw[<-] (t_1) -- (v_6);
	\draw[<-] (t_2) -- (v_7);
	\draw[<-] (t_3) -- (v_8);

	\node(v_9)  at (5, {1 / sqrt(2)}) [label = above:$v_9$] {};
	\node(v_10) at (5, 0)             [label = above:$v_{10}$] {};
	\draw[->] (v_6) -- (v_9);
	\draw[->] (v_7) -- (v_9);
	\draw[->] (v_8) -- (v_10);

	\node(v_11)  at (6, {1 / sqrt(2)}) [label = above:$v_{11}$] {};
	\draw[<-] (v_9) --  (v_11);
	\draw[<-] (v_10) -- (v_11);

	\node(v_12)  at 
	    ({5 + 1 / sqrt(2) + cos(-115)}, 
             {  - 1 / sqrt(2) + sin(-115)}) 
	     [label = left:$v_{12}$] {};
	\node(v_13)  at 
	    ({5 + 1 / sqrt(2) + cos(15)}, 
             {  - 1 / sqrt(2) + sin(15)}) 
	     [label = above:$v_{13}$] {};
	\node(v_14)  at ({5 + 1 / sqrt(2)}, { - 1 / sqrt(2)}) [label = above:$v_{14}$] {};
	\node(v_15)  at ({5 + 2 / sqrt(2)}, { - 2 / sqrt(2)}) [label = right:$v_{15}$] {};

	\draw[<-] (v_10) -- (v_12);
	\draw[<-] (v_10) -- (v_13);
	\draw[->] (v_12) -- (v_14);
	\draw[->] (v_12) -- (v_15);
	\draw[->] (v_13) -- (v_14);
	\draw[->] (v_13) -- (v_15);
	\draw[->] (v_14) -- (v_15);
\end{tikzpicture}
  \caption{$H$}
\end{figure}
The orientation of the edges is transitive, hence $H$ is a comparability graph.
Further, comparability graphs are perfect, hence $\chi(H) = \omega(H) = 3$,
i.e. $H$ is three-colourable.

\begin{lemma}
  \label{H_lemma}
  Exactly one of $t_1$, $t_2$ and $t_3$ is right
  in every monopolar partition of $H$ .
  There are exactly three monopolar partitions of $H$
  and each one is uniquely determined by which vertex of $\{t_1, t_2, t_3\}$ is right.
\end{lemma}
\begin{proof}
  Let $V_{Q_1} = \{v_1, \ldots v_5\}$ and
  $V_{Q_2} = \{v_{10}, v_{12}, v_{13}, v_{14}, v_{15}\}$.
  Observe that $H[V_{Q_1}] \cong H[V_{Q_2}] \cong Q$,
  and therefore in every monopolar partition
  $v_5$ and $v_{10}$ are right vertices and they cannot have
  other right neighbours from $V_{Q_1}$ and $V_{Q_2}$ respectively.
  The vertex $v_5$ must be right, so at most one of $\{t_1, t_2, t_3\}$ can be right.
  It is a routine to check that
  setting all $t_1$, $t_2$ and $t_3$ left cannot yield a monopolar partition.
  It is also a routine to check that setting each one of $\{t_1, t_2, t_3\}$ right
  and the other two left defines a unique monopolar partition.
\end{proof}

We associate every clause of an arbitrary positive 3-CNF formula $\phi$
with an independent copy of $H$.
The selection of the right vertex among $\{t_1, t_2, t_3\}$ in a monopolar partition
will indicate which of the variables of the clause is true.
Then we add extra synchronisation vertices,
each uniquely associated with a variable of $\phi$
and joined by an edge to all vertices associated with the same variable
in the copies of $H$.

More formally,
let $\phi = \bigwedge_{i=1}^mC_i$ be a positive $3$-CNF formula
with variables $c_1 \ldots c_n$.
Let $G$ be the disjoint union of an independent set $\{x_1, \ldots x_n\}$
and $m$ disjoint copies of $H$, $\{H_i\}_{i=1}^m$.
For each clause $C_i$, say $C_i = \{c_k \lor c_l \lor c_p\}$,
connect $x_k$ with $t_{i, 1}$, $x_l$ with $t_{i, 2}$ and $x_p$ with $t_{i, 3}$,
where $t_{i, 1}$, $t_{i, 2}$ and $t_{i, 3}$ are respectively
$t_1$, $t_2$ and $t_3$ in $H_i$.

\begin{define}
  For every positive 3-CNF formula $\phi$ define $G_\phi$ to be the graph described above.
\end{define}

\begin{lemma}
  \label{lemma_transform}
  A positive 3-CNF formula $\phi$
  is a ``yes''-instance of 1-3-SAT
  iff $G_\phi$ is monopolar.
\end{lemma}
\begin{proof}
  $(\Rightarrow)$
  Assume that $f: \{c_i\} \rightarrow \{true, false\}$ is an assignment of the
  variables of $\phi$
  such that every clause contains exactly one true literal.
  Create a partition $(A, B)$ of $V(G_\phi)$ as follows:
  $x_i \in A$ for each $c_i$ with $f(c_i) = true$, and $x_i \in B$ otherwise.
  The vertices $\{x_1 \ldots x_n\}$ have disjoint neighbourhoods,
  so we can extend the partition above with $v \in A \Leftrightarrow x_i \in B$
  for each $x_i$ and neighbour $v$ of $x_i$.
  By Lemma \ref{H_lemma} the partition above can be
  further extended uniquely to each $H_i$,
  and hence the entire graph $G_\phi$.
  We conclude that $(A, B)$ is a monopolar partition of $G_\phi$.

  $(\Leftarrow)$
  Let $(A, B)$ be a monopolar partition of $G_\phi$.
  Define $f : \{c_i\} \rightarrow \{true, false\}$
  as follows: $f(c_i) = true \leftrightarrow x_i \in A$.
  From Lemma \ref{H_lemma} for each clause $C_i$ there is a unique
  $t_{i, j} \in B$, which is adjacent to $v_{i, 5} \in B$.
  Observe that $v_{i, 5}$ is non-adjacent to any $x_j$, hence
  for each clause $C$ there is a unique $x_j \in A$
  with $c_j \in C$.
\end{proof}

\begin{lemma}
  For every positive 3-CNF formula $\phi$,
  $G_\phi$ is a 3-colourable comparability graph.
\end{lemma}
\begin{proof}
  To see that $G_\phi$ is a comparability graph,
  orient each copy of $H$ as oriented in Figure~2.
  Orient the remaining edges from $x_1 \ldots x_n$ towards $t_{i, j}$.
  The described orientation of $G_\phi$ is transitive.
  As $G_\phi$ is perfect, we have $\chi(G_\phi) = \omega(G_\phi) = 3$.
\end{proof}

\begin{cor}
  \label{monopolar_final}
  Monopolarity is NP-complete even if the input graph is restricted to be
  $3$-colourable comparability graph.
\end{cor}
\begin{proof}
  Monopolarity is in NP regardless of the restrictions on the input graph.
  Furthermore,
  NP-hardness follows from the reduction of 1-3-SAT in the statement
  of Lemma \ref{lemma_transform}.
\end{proof}

It is worth noting that further restrictions
can be imposed on problem.
We can design $G_\phi$ so that the different copies of $H$
share the four triangles they contain,
and therefore build $G_\phi$ with a constant number (four) of triangles.


\section{Polarity}
It this section we prove that it is NP-complete to decide whether a
3-colourable comparability graph is polar or not.
We do this by showing that polarity is at least as hard as monopolarity
for classes of graphs which are closed under disjoint union.
Note that $3$-colourable comparability graphs form such class.

\begin{claim}
  \label{disconencted_claim}
  Suppose $G$ is a complement of a union of disjoint cliques.
  Then $G$ is either connected or empty.
\end{claim}
\begin{proof}
  Since $\overline{G}$ is a union of disjoint cliques,
  $\overline{G}$ is either a clique
  of disconnected.
  The complement of a disconnected graph is connected,
  therefore $G$ is either empty or connected.
\end{proof}

We use the notation $G = 2H$ to express that $G$ is a union of two disjoint copies of $H$
without edges inbetween.

\begin{lemma}
  \label{transform}
  The following three statements are equivalent:
  \begin{enumerate}
    \item      $H = 2G$ is polar.
    \item      $G$ is monopolar
    \item      $H = 2G$ is monopolar.
  \end{enumerate}
\end{lemma}
\begin{proof}
    $(2. \Rightarrow 3.)$ Trivial.

    $(3. \Rightarrow 1.)$ Trivial.

    $(1. \Rightarrow 2.)$
    Let $(A, B)$ be a polar partition of $V(H)$,
    and let $(V_1, V_2)$ be a partition of $V(H)$,
    such that $H[V_i] \cong G$.
    Let $A_i = A \cup V_i$.
    If $A_1$ or $A_2$ is empty, then $G$ is a union of cliques
    and hence monopolar.
    Otherwise, $H[A]$ is disconnected because there are no edges between $A_1$ and $A_2$.
    But $H[A]$ is a co-union of cliques,
    hence $H[A]$ is empty by Claim \ref{disconencted_claim}.
    We deduce that $(A_i, V_i \setminus A_i)$ is a monopolar partition $G[V_i]$.
\end{proof}

\begin{lemma}
  \label{hardness}
  Let $\mathcal{P}$ be a class of graphs which is closed under disjoint union.
  Determining polarity of instances restricted to $\mathcal{P}$
  is at least as hard as determining monopolarity
  for the same class of instances.
\end{lemma}
\begin{proof}
  We reduce monopolarity for input graphs restricted to $\mathcal{P}$
  to polarity for input graphs restricted to $\mathcal{P}$.
  To decide if $G \in P$ is monopolar
  it is sufficient to decide whether $2G \in \mathcal{P}$ is polar
  by Lemma \ref{transform}.
\end{proof}

\begin{cor}
  \label{cor_hardness}
  The problem of deciding whether a $3$-colourable comparability graph is polar
  is an NP-complete problem.
\end{cor}
\begin{proof}
  The problem is clearly in NP,
  and it is NP-complete to decide whether such graph is monopolar
  by Corollary~\ref{monopolar_final},
  hence the statement follows from Lemma \ref{hardness}.
\end{proof}

\begin{cor}
  It is an NP-complete problem to decide whether a co-comparability graph is polar.
\end{cor}
\begin{proof}
  A graph is polar if its complement is polar.
  In order to decide whether a comparability graph is polar,
  it is sufficient to decide whether its complement, a co-comparability graph,
  is polar.
  The prior decision problem is NP-complete by Corollary \ref{cor_hardness},
  hence the latter is also NP-complete.
\end{proof}

\bibliographystyle{alpha}
\bibliography{main}

\end{document}